\documentclass{article}
\usepackage[numbers]{natbib}
\usepackage{amsthm}
\usepackage{amsmath}
\usepackage{amssymb}
\usepackage{stmaryrd}
\usepackage{semantic}
\usepackage{hyperref}

\newtheorem{theorem}{Theorem}
\newtheorem{lemma}[theorem]{Lemma}

\newtheorem{proposition}[theorem]{Proposition}

\newtheorem{definition}[theorem]{Definition}

\title{Transitivity of Subtyping for Intersection Types}
\author{Jeremy G. Siek}

\begin{document}
\maketitle

\newcommand{\TOP}{\ensuremath{\mathtt{U}}}
\newcommand{\dom}[1]{\cap\mathsf{dom}(#1)}
\newcommand{\cod}[1]{\cap\mathsf{cod}(#1)}
\newcommand{\topP}[1]{\mathsf{top}(#1)}
\newcommand{\topInCod}[1]{\mathsf{topInCod}(#1)}
\newcommand{\depth}[1]{\mathsf{depth}(#1)}
\newcommand{\size}[1]{\mathsf{size}(#1)}
\newcommand{\inside}[0]{\inplus}
\newcommand{\containedin}[0]{\subsetpluseq}

\begin{abstract}
The subtyping rules for intersection types traditionally employ a
transitivity rule (Barendregt et al. 1983), which means that
subtyping does not satisfy the subformula property, making it more
difficult to use in filter models for compiler verification.
Laurent develops a sequent-style subtyping system, without
transitivity, and proves transitivity via a sequence of six lemmas
that culminate in cut-elimination (2018). This article develops a
subtyping system in regular style that omits transitivity and
provides a direct proof of transitivity, significantly reducing the
length of the proof, exchanging the six lemmas for just
one. Inspired by Laurent's system, the rule for function types is
essentially the $\beta$-soundness property.  The new system
satisfies the ``subformula conjunction property'': every type
occurring in the derivation of $A <: B$ is a subformula of $A$ or
$B$, or an intersection of such subformulas. The article proves that
the new subtyping system is equivalent to that of Barendregt, Coppo,
and Dezani-Ciancaglini.
\end{abstract}

\section{Introduction}

Intersection types were invented by Coppo, Dezani-Ciancaglini, and
Salle to study normalization in the lambda
calculus~\citep{Coppo:1979aa}. Subsequently intersection types have
been used for at least three purposes:
\begin{enumerate}
  \item in type systems~\citep{Reynolds:1988aa,Pierce:1991aa,Castagna:2014aa,Chaudhuri:2014aa,Oliveira:2016aa,Amin:2017aa,Muehlboeck:2018aa,Bi:2019aa,Dunfield:2019aa,Microsoft:TypeScript2020aa},
  \item in precise static
    analyses~\citep{Turbak:1997aa,Palsberg:1998aa,Mossin:2003aa,Simoes:2007aa},
    and
  \item in the denotational semantics for a variety of lambda
    calculi~\citep{Coppo:1980ab,Coppo:1981aa,Coppo:1984aa,Honsell:1992aa,Abramsky:1993fk,Honsell:1999aa,Ishihara:2002aa,Rocca:2004aa,Dezani-Ciancaglini:2005aa,Alessi:2006aa,Barendregt:2013aa}.
\end{enumerate}
The motivation for this article comes from the third use of
intersection types, in denotational semantics. We are interested in
constructing filter models for the denotational semantics of
functional programming languages in support of mechanized compiler
verification. Section~\ref{sec:motivation} describes how that research
motivates the results in this article.

There are many intersection type systems; perhaps the best-known them
is the BCD intersection type system of \citet{Barendregt:1983aa}. For
this article we focus on the BCD system, following the presentation of
\citet{Barendregt:2013aa}. We conjecture that our results apply to
other intersection type systems that include the $({\to}{\cap})$ rule
for distributing intersection and function types.

The BCD intersection type system includes function types, $A \to B$,
intersection types, $A \cap B$, a top type, $\TOP$, and an infinite
collection of type constants. Figure~\ref{fig:types} defines the
grammar of types.

\begin{figure}[tbp]
  \[
  \begin{array}{lclr}
    \alpha,\beta & ::= & \TOP \mid c_0 \mid c_1 \mid c_2 \mid \cdots & \text{atoms}\\
    A,B,C,D & ::= & \alpha \mid A \to B \mid A \cap B & \text{types}
  \end{array}
  \]
  \caption{Intersection Types}
  \label{fig:types}
\end{figure}

The BCD intersection type system includes a subsumption rule which
states that a term $M$ in environment $\Gamma$ can be given type $B$
if it has type $A$ and $A$ is a subtype of $B$, written $A \leq B$.
\[
\inference{\Gamma \vdash M : A & A \leq B}
          {\Gamma \vdash M : B}
\]
Figure~\ref{fig:BCD-subtyping} reviews the BCD subtyping system.  Note
that in the (trans) rule, the type $B$ appears in the premises but not
in the conclusion. Thus, the BCD subtyping system does not satisfy the
subformula property.  For many other subtyping systems, it is
straightforward to remove the (trans) rule, modify the other rules,
and then prove transitivity.  Unfortunately, the $({\to}{\cap})$ rule
of the BCD system significantly complicates the situation.  One might
hope to omit the $({\to}{\cap})$ rule, but it plays an import role in
ensuring that the filter model produces a compositional semantics.

\begin{figure}[tbp]
  \fbox{$A \leq B$}
  \begin{gather*}
    \text{(refl)} \; \inference{}{A \leq A} \quad
    \text{(trans)} \; \inference{A \leq B & B \leq C}{A \leq C} \\[3ex]
    \text{(incl$_L$)} \; \inference{}{A \cap B \leq A} \quad
    \text{(incl$_R$)} \; \inference{}{A \cap B \leq B} \quad
    \text{(glb)} \; \inference{A \leq C & A \leq D}{A \leq C \cap D} \\[3ex]
    (\to) \; \inference{C \leq A & B \leq D}{A \to B \leq C \to D} \quad
    ({\to}{\cap}) \; \inference{}{(A \to B) \cap (A \to C) \leq A \to (B \cap C)} \\[3ex]
    (\TOP_{\mathrm{top}}) \; \inference{}{A \leq \TOP} \quad
    (\TOP{\to}) \; \inference{}{\TOP \leq C \to \TOP}
  \end{gather*}
  \caption{Subtyping of Barendregt, Coppo, and
    Dezani-Ciancaglini (BCD).}
  \label{fig:BCD-subtyping}
\end{figure}

\section{Motivations from Verified Compilation}
\label{sec:motivation}

Our interest in intersection type systems stems from our work on
mechanized correctness proofs for compilers for functional programming
languages. We are exploring whether denotational semantics can
streamline such proofs compared to the operational techniques such as
those used to prove the correctness of
CakeML~\citep{Kumar:2014aa,Owens:2017aa}. In particular, we are
constructing filter models because they are relatively straightforward
to mechanize in a proof assistant compared to traditional domain
theory~\citep{Benton:2009ab,Dockins:2014aa}.  In a filter model, the
meaning of a program $M$ is the set of all types that can be assigned
to it by an intersection type system:
\[
   \llbracket M \rrbracket = \{ A \mid \emptyset \vdash M : A \}
\]
Most type systems are much too coarse of an abstraction to be used as
filter models. Intersection type systems are special in that they can
completely characterize the runtime behavior of a program.  To catch a
glimpse of how filter models work, consider the meaning of the
identity function in the $\lambda$-calculus.
\[
  \llbracket \lambda x.\, x \rrbracket =
      \{ c_0 \to c_0, c_1 \to c_1, c_2 \to c_2, \ldots \}
\]
If one takes each type constant $c_i$ to represent a singleton type
for the number $i$, then the set of types assigned to the identity
function starts to look like the graph of the function. Indeed, filter
models are closely related to graph
models~\citep{Scott:1976lq,Engeler:1981aa,Plotkin:1993ab,Barendregt:2013aa}.

One of the standard compiler transformations that we are interested in
verifying is common subexpression
elimination~\citep{Downey:JACM:1980,Appel:1992fk,Tarditi:1996aa}. For
example, this transformation would replace the duplicated $(f \; x)$
terms in the following
\[
  M = \lambda f.\, \lambda h.\, \lambda x.\, h \; (f \; x) \; (f \; x)
\]
with a single instance of $(f\;x)$ that is let-bound to a variable $y$:
\[
M' = \lambda f.\, \lambda h.\, \lambda x.\,
\mathsf{let}\; y = (f \; x) \;\mathsf{in}\;
h \; y \; y
\]
Unfortunately, $M$ and $M'$ are not denotationally equal according to
most filter models because filter models allow function graphs that
represent arbitrary relations. The above transformation depends on $(f
\; x)$ producing the same result every time, but a filter model
associates arbitrary graphs with parameter $f$.

To solve this problem, one needs to restrict the graphs to only allow
approximations of functions. This can be accomplished by requiring
every entry in a function's graph to be ``consistent''. Two entries
are consistent, written $(A \to B) \sim (A' \to B')$, when either $A
\not\sim A'$ or $A \sim A'$ and $B \sim B'$. Two type constants are
consistent when they are equal. One would then hope to prove that for
any term $M$, if $\Gamma \vdash M : C$ and $\Gamma \vdash M : C'$,
then $C \sim C'$ as long as $C$, $C'$, and $\Gamma$ are self
consistent. A key lemma needed for that proof is that consistency is
upward closed with respect to subtyping: if $A \sim B$, $A <: C$, $B
<: D$, then $C \sim D$.

This brings us back to the transitivity rule and the subformula
property. The proof of the above lemma hits a snag in the cases where
$A <: C$ and $B <: D$ are derived using the transitivity rule. The
intermediate type is not guaranteed to be self consistent. We explored
adding this guarantee to the transitivity rule, but it caused a
significant increase in obligations elsewhere in our proofs. By some
serendipity, around that time \citet{Laurent:2018aa} published a
subtyping system that removes the transitivity rule and satisfies the
subformula property.  Laurent developed a sequent-style subtyping
system, written $\Gamma \vdash B$, where $\Gamma$ is a sequence of
types $A_1,\ldots,A_n$. The intuition is that $A_1,\ldots,A_n \vdash
B$ corresponds to $A_1 \cap \cdots \cap A_n \leq B$. This system
satisfies the subformula property and is equivalent to the BCD
system. To prove this, Laurent establishes six lemmas that culminate
in cut-elimination, from which transitivity follows. Laurent
mechanized these proofs in Coq.

We immediately used Laurent's result to complete a filter model for
ISWIM~\citep{Landin:1966la,G.-D.-Plotkin:1975on,Felleisen:2009aa} and
mechanized the results in the Isabelle proof
assistant~\citep{Siek:2018aa}.

A year later we began to port this filter model to the Agda proof
assistant. We needed several variations on the filter model to give
semantics to different languages, including ones that are
call-by-value and call-by-name, and to the different intermediate
languages used in our compiler.  The availability of dependent types
in Agda promised to make our development more reusable.  However, we
still found it tedious to repeat Laurent's proof of transitivity for
each of the filter models, and in the process we developed an
intuition that the sequent-style system and the six lemmas were not
necessary.

The present article discusses how we were able to import the key
feature from Laurent's system, the rule for subtyping between function
types, into the BCD subtyping system, replacing the function $(\to)$,
distributivity $({\to}{\cap})$, and transitivity (trans) rules.  We
then prove transitivity without the sequence of six lemmas, instead
using one lemma that has been used previously to prove the inversion
principle for subtyping of function
types~\citep{Barendregt:2013aa}. The new subtyping system does not
technically satisfy the subformula property but it satisfies what we
call the \emph{subformula conjunction property}: every type occurring
in the derivation of $A <: B$ is a subformula of $A$ or $B$ or an
intersection of such subformulas.  This property is enough to ensure
that if $A$ and $B$ are self consistent, then so are all the types
that appear in the derivation of $A <: B$.  Thus, the new subtyping
system is suitable for mechanizing filter models.

\section{Road Map}
\label{sec:road-map}

The rest of this article is organized as follows. We discuss two
recent developments regarding subtyping for intersection types in
Section~\ref{sec:recent-developments}.  We present our new subtyping
system in Section~\ref{sec:new-subtyping}, prove transitivity in
Section~\ref{sec:trans}, and prove its equivalence to BCD subtyping in
Section~\ref{sec:equiv}. We conclude the article in
Section~\ref{sec:conclude}.

The definitions and results in this article have been machine checked
in Agda, in the file \texttt{agda/TransSubInter.agda} in the following
repository.
\begin{center}
  \url{https://github.com/jsiek/denotational_semantics}
\end{center}

\section{Recent Developments}
\label{sec:recent-developments}

There has been a recent flurry of interest in subtyping for
intersection types. We begin by discussing two algorithms that
effectively satisfy the subformula property. We then discuss issues
surrounding subtyping and transitivity in the Scala language.

\citet{Bi:2018aa} present a subtyping algorithm that satisfies the
subformula property and proves that it is equivalent to BCD
subtyping. Their system is based on the decision procedure
of~\citet{Pierce:1989aa} and takes the form $\mathcal{L} \vdash A
\prec: B$, where $\mathcal{L}$ is a queue of types that are peeled off
from the domain of the type on the right. So the judgment
$\mathcal{L} \vdash A \prec: B$ is equivalent to $A \leq \mathcal{L}
\to B$ where $\mathcal{L} \to B$ is defined as
\begin{align*}
  [] \to B &= B\\
  (\mathcal{L},A) \to B &= \mathcal{L} \to (A \to B)
\end{align*}
The proof of transitivity for this system is an adaptation of Pierce's
but it corrects some errors and introduces more lemmas concerning an
auxiliary notion called reflexive supertypes.

\citet{Muehlboeck:2018aa} develop a framework for obtaining subtyping
algorithms for systems that include intersection and union types and
that are extensible to other types. To decide $A \leq B$ their
algorithm converts $A$ to disjunctive normal form and then applies a
client-supplied function to each collection of literals. To obtain a
system equivalent to BCD, including the ${\to}{\cap}$ rule, the client
side function saturates the collection of literals by applying
${\to}{\cap}$ (left to right) as much as possible.

These two algorithms were developed for the purposes of type checking
programming languages, and not for the construction of filter models,
so the goals are somewhat different. In short, the algorithms
introduce complications for the sake of efficiency, which would be
undesirable for use in a filter model.

Research on the Scala programming language~\citep{Odersky:2004aa} has
led to the development of the Dependent Object Types (DOT) calculus,
which includes path-dependent types and intersection
types~\citep{Amin:2016aa,Amin:2017aa}.  The interplay between
path-dependent types, subtyping, and transitivity in DOT has proved a
challenge in the proof of type soundness. Nevertheless,
\citet{Rompf:2016aa} prove the inversion principle for function types
using a ``pushback'' lemma that reorganizes any subtyping derivation
such that the last rule to be applied is never transitivity. Applying
such a lemma exhaustively might provide an alternative route to
transitivity for the new subtyping system presented in this
article. Conversely, the new subtyping system might enable the
addition of distributivity of function and intersection types to
DOT. However, we suspect that the pushback approach would encounter
difficulties if applied to the original BCD subtyping rules.

\section{A New Subtyping System}
\label{sec:new-subtyping}

Our new subtyping system relies on a few definitions that are given in
Figure~\ref{fig:aux}. These include the partial functions $\dom{A}$
and $\cod{A}$, the $\topP{A}$ and $\topInCod{A}$ predicates, and the
relations $A \inside B$ and $A \containedin B$.
The $\dom{A}$ and $\cod{A}$ partial functions return the domain or
codomain if $A$ is a function type, respectively. If $A$ is an
intersection $A_1 \cap A_2$, then $\dom{A}$ is the intersection of the
domain of $A_1$ and $A_2$.  If $A$ is an atom, $\dom{A}$ is
undefined. Likewise for $\cod{A}$. For example, if $A = (A_1 \to B_1)
\cap \cdots \cap (A_n \to B_n)$, then $\dom{A} = A_1 \cap \cdots \cap
A_n$ and $\cod{A} = B_1 \cap \cdots \cap B_n$.  When $\dom{A}$ or
$\cod{A}$ appears in lemma or theorem statement, we implicitly assume
that $A$ is a type such that $\dom{A}$ and $\cod{A}$ are defined.
The $\topP{A}$ predicate identifies types that are equivalent to
$\TOP$. (See Proposition~\ref{prop:top}(\ref{prop:AllBot-⊑-any}).) The
$\topInCod{A}$ predicate identifies types that have $\TOP$ in their
codomain.
The relations $A \inside B$ and $A \containedin B$ enable the
treatment of a sequence of intersections $B_1 \cap \cdots \cap B_n$
(modulo associativity and commutativity) as if it were a set of types
$\{ B_1, \ldots, B_n \}$, where each $B_i$ is an atom or a function
type.  We say that $A$ is a part of $B$ if $A \inside B$ and we say
that $B$ \emph{contains} $A$ if $A \containedin B$.

\begin{proposition}\label{prop:union-subset-inv}
 \item If $A \cap B \containedin C$, then $A \containedin C$ and $B \containedin C$. 
\end{proposition}

\begin{figure}[tbp]

  \fbox{$\dom{A}, \cod{A}$}
  \begin{align*}
  \dom{A \to B} &= A \\
  \dom{A \cap B} &= \dom{A} \,\cap\, \dom {B} \\
  \\
  \cod{A \to B} &= B \\
  \cod{A \cap B} &= \cod{A} \,\cap\, \cod {B}
  \end{align*}

  \fbox{$A \inside B$}
  \begin{gather*}
    \inference{}{\alpha \inside \alpha}  \quad
    \inference{}{A \to B \inside A \to B} \quad
    \inference{A \inside B}{A \inside B \cap C} \quad
    \inference{A \inside C}{A \inside B \cap C}
  \end{gather*}

  \fbox{$A \containedin B$}
  \[
     A \containedin B = \forall C.\, C \inside A \text{ implies } C \inside B
  \]

  \fbox{$\topP{A}$}
  \begin{gather*}
    \inference{}{\topP{\TOP}}
    \quad
    \inference{\topP{B}}{\topP{A \to B}}
    \quad
    \inference{\topP{A} & \topP{B}}{\topP{A \cap B}}
  \end{gather*}

  \fbox{$\mathsf{topInCod}(D)$}
  \[
  \mathsf{topInCod}(D) =
     \exists A B.\, A \to B \inside D \text{ and } \mathsf{top}(B)  
  \]

  \caption{Auxiliary Definitions}
  \label{fig:aux}
\end{figure}

\begin{figure}[tbp]
  \fbox{$A <: B$}
  \begin{gather*}
    \text{(refl$_\alpha$)} \; \inference{}{\alpha <: \alpha} \\[3ex]
    \text{(lb$_L$)} \; \inference{A <: C}{A \cap B <: C} \quad
    \text{(lb$_R$)} \; \inference{B <: C}{A \cap B <: C} \quad
    \text{(glb)} \; \inference{A <: C & A <: D}{A <: C \cap D} \\[3ex]
    (\to') \; \inference{C <: \dom{B} & \cod{B} <: D }{A <: C \to D}
    \begin{array}{l}
      B \containedin A\\
      \neg\, \mathsf{top}(D) \\
      \neg \, \mathsf{topInCod}(B)
      \end{array}\\[3ex]
    (\TOP_{\mathrm{top}}) \; \inference{}{A <: \TOP} \quad
    (\TOP{\to}') \; \inference{}{A <: C \to D}\;\mathsf{top}(D)
  \end{gather*}
  \caption{The New Subtyping System}
  \label{fig:new-subtyping}
\end{figure}

The new intersection subtyping system, with judgments of the form $A
<: B$, is defined in Figure~\ref{fig:new-subtyping}. First, it does
not include the (trans) rule.  It also replaces the (refl) rule with
reflexivity for atoms (refl$_\alpha$). The most important rule is the
one for function types $(\to')$, which subsumes $(\to)$ and
$({\to}{\cap})$ in BCD subtyping.  The $(\to')$ rule essentially turns
the $\beta$-soundness property into a subtyping
rule~\citep{Barendregt:2013aa}. The $(\to')$ rule says that a type $A$
is a subtype of a function type $C \to D$ if a type contained in $A$,
call it $B$, has domain and codomain that are larger and smaller than
$C$ and $D$, respectively. The use of containment enables this rule to
absorb uses of (incl$_L$) and (incl$_R$) on the left.  The side
conditions $\neg\;\topP{B}$ and $\neg\;\topInCod{D}$ are needed
because of the $(\TOP{\to}')$ rule, which in turn is needed to
preserve types under $\eta$-reduction.  In the many systems that do
not involve $\eta$-reduction, the $(\TOP{\to}')$ rule can be omitted,
as well as these side conditions. The rules (lb$_L$) and (lb$_R$)
adapt (incl$_L$) and (incl$_R$) to a system without transitivity, and
have appeared many times in the literature~\citep{Bakel:1995aa}.  The
$(\TOP{\to}')$ rule generalizes the $(\TOP{\to})$ rule, replacing the
$\TOP$ on the left with any type $A$, because for transitivity, any
type is below $\TOP$. The $(\TOP{\to}')$ rule replaces the $\TOP$ in
the codomain on the right with any type $D$ that is equivalent to
$\TOP$.

Before moving on, we take note of some basic facts regarding the $<:$
relation and the $\topP{A}$ predicate.

\begin{proposition}[Basic Properties of $<:$]\ \label{prop:subtyping}
  \begin{enumerate}
  \item (reflexivity) $A <: A$ \label{prop:⊑-refl}
  \item If $A <: B \cap C$, then $A <: B$ and $A <: C$. \label{prop:⊔⊑-inv}
  \item If $A <: B$ and $C \inside B$, then $A <: C$.\label{prop:in-sub-sub}
  \item If $A <: B$ and $C \containedin B$, then $A <: C$.\label{prop:subset-sub-sub}
  \end{enumerate}
\end{proposition}
\begin{proof}\ 
  \begin{enumerate}
  \item The proof of reflexivity is by induction on $A$. In the case
    $A = A_1 \to A_2$, we proceed by cases on whether $\topP{A_2}$.
    If it is, deduce $A_1 \to A_2 <: A_1 \to A_2$ by rule $(\TOP{\to}')$.
    Otherwise, apply rule $(\to')$
  \item The proof is by induction on the derivation of $A <: B \cap C$.
  \item The proof is by induction on $B$. In the case where $B = B_1
    \cap B_2$, either $C \inside B_1$ or $C \inside B_2$, but in either case part
    \ref{prop:⊔⊑-inv} of this proposition fulfills the premise of
    the induction hypothesis, from which the conclusion follows.
  \item The proof  is by induction on $C$, using part \ref{prop:in-sub-sub}
    of this proposition in the cases for atoms and function types.
  \end{enumerate}
\end{proof}

\begin{proposition}[Properties of $\topP{A}$]\ \label{prop:top}
  \begin{enumerate}
  \item If $\topP{A}$ then $\topP{\cod{A}}$.\label{prop:AllBot-cod}
  \item If $\topP{A}$ and $B \inside A$, then $\topP{B}$.\label{prop:AllBot-in}
  \item If $\topP{A}$ and $B \containedin A$, then $\topP{B}$.\label{prop:AllBot-subset}
  \item If $\topP{A}$ and $A <: B$, then $\topP{B}$.\label{prop:AllBot-⊑}
  \item If $\topP{A}$, then $B <: A$.\label{prop:AllBot-⊑-any}
  \end{enumerate}
\end{proposition}
\begin{proof}\
  \begin{enumerate}
  \item The proof is a straightforward induction on $A$.
  \item The proof is also a straightforward induction on $A$.
  \item The proof is by induction on $B$. The cases for atoms and
    function types are proved by part \ref{prop:AllBot-in} of
    this proposition. In the case for
    $B = B_1 \cap B_2$, from $B_1 \cap B_2 \containedin A$, we have
    $B_1 \containedin A$ and $B_2 \containedin A$
    (Proposition~\ref{prop:union-subset-inv}).
    Then by the induction hypotheses for $B_1$ and $B_2$ we have
    $\topP{B_1}$ and $\topP{B_2}$, from which we conclude
    that $\topP{B_1 \cap B_2}$.
  \item The proof is by induction on the derivation of $A <: B$.
    All of the cases are straightforward except for rule $(\to')$.
    In that case we have $B = B_1 \to B_2$ and some $A'$ such that
    $A' \containedin A$, $B_1 <: \dom{A'}$, $\cod{A'} <: B_2$, $\neg\,\topP{B_2}$,
    and $\neg\,\topInCod{A'}$. From the premise $\topP{A}$
    and part \ref{prop:AllBot-subset} of this proposition, we
    have $\topP{A'}$. Then by part \ref{prop:AllBot-cod}
    we have $\topP{\cod{A'}}$. By the induction hypothesis
    for $\cod{A'} <: B$ we conclude that $\topP{B}$.
  \item The proof is a straightforward induction on $A$.
  \end{enumerate}
\end{proof}

Next we turn to the subtyping inversion principle for function types.
The idea is to generalize the rule $(\to')$ with respect to the type on
the right, allowing any type that contains a function type.  The
premises of $(\to')$ are somewhat complex, so we package most of them
into the following definition.

\begin{definition}[factors]
  We say $C \to D$ \emph{factors} $A$
  if there exists some type $B$ such that
  $B \containedin A$, $C <: \dom{B}$, $\cod{B} <: D$, and
  $\neg\,\mathsf{topInCod}(B)$.
\end{definition}

\begin{proposition}[Inversion Principle for Function Types]
  \label{prop:⊑-fun-inv}\ \\
  If $A <: B$, $C \to D \inside B$, and $\neg\,\mathsf{top}(D)$, then
  $C \to D$ factors $A$.
\end{proposition}
\begin{proof}
  The proof is a straightforward induction on $A <: B$.
\end{proof}

\section{Proof of Transitivity of Subtyping}
\label{sec:trans}

Our proof of transitivity relies on the following lemma, which is
traditionally needed to prove the inversion principle for function
types. (In our system this lemma is not needed to prove the inversion
principle because the rule $(\to')$ is already quite close to the
inversion principle.)  The lemma states that if every function type $C
\to D$ in A factors $B$, then $\dom{A} \to \cod{A}$ also factors $B$.

\begin{lemma}\label{lem:sub-inv-trans}
  If
  \begin{itemize}
  \item for any $C$ $D$, if $C \to D \inside A$ and $\neg\,\mathsf{top}(D)$,
    then $C \to D$ factors $B$, and
  \item $\neg\, \mathsf{topInCod}(A)$,
  \end{itemize}
  then $\dom{A} \to \cod{A}$ factors $B$.
\end{lemma}
\begin{proof}
  The proof is by induction on $A$.
  \begin{itemize}
  \item Case $A$ is an atom. The statement is vacuously true.
  \item Case $A = A_1 \to A_2$ is a function type. Then we conclude by applying
    the premise with $C$ and $D$ instantiated to $A_1$ and $A_2$ respectively.
  \item Case $A = A_1 \cap A_2$.  By the induction hypothesis for $A_1$
    and for $A_2$, we have that $\dom{A_1} \to \cod{A_1}$ factors $B$
    and so does $\dom{A_2} \to \cod{A_2}$.  So there exists
    $B_1$ and $B_2$ such that $B_1 \containedin B$, $\neg\,\topInCod{B_1}$,
    $\dom{A_1} <: \dom{B_1}$, $\cod{B_1} <: \cod{A_2}$ and similarly
    for $B_2$. We need to show that $\dom{A} \to \cod{A}$ factors
    $B$. We choose the witness $B_1 \cap B_2$.  Clearly we have $B_1 \cap
    B_2 \containedin B$ and $\neg \topInCod{B_1 \cap B_2}$.  Also, we have
    \[
    \dom{B_1} \cap \dom{B_2} <: \dom{A_1} \cap \dom{A_2}
    \]
    and
    \[
    \cod{A_1} \cap \cod{A_2} <: \cod{B_1} \cap \cod{B_2}
    \]
    Thus, we have that $\dom{B_1 \cap B_2} <: \dom{A}$
    and $\cod{A} <: \cod{B_1 \cap B_2}$, and this case is complete.
  \end{itemize}
\end{proof}

We now turn to the proof of transitivity, that if $A <: B$ and $B <:
C$, then $A <: C$.
The proof is by well-founded induction on the lexicographical ordering
of the depth of $B$, the size of $B$, and then the size of $C$.  To be
precise, we define this ordering as follows.
\[
\langle A, B, C \rangle \ll \langle A', B', C' \rangle
=
\begin{array}{l}
\depth{B} < \depth{B'} \\
\text{ or } \depth{B} \leq \depth{B'} \text{ and } \size{B} < \size{B'} \\
\text{ or }
  \depth{B} \leq \depth{B'} \text{ and } \size{B} \leq \size{B'}\\
  \;\quad \text{ and } \size{C} < \size{C'}
\end{array}
\]
%
where $\size{A}$ is 
\begin{align*}
  \size{\alpha} &= 0 \\
  \size{A \to B} &= 1 + \size{A} + \size{B} \\
  \size{A \cap B} &= 1 + \size{A} + \size{B}
\end{align*}
and $\depth{A}$ is 
\begin{align*}
  \depth{\alpha} &= 0 \\
  \depth{A \to B} &= 1 + \max(\depth{A}, \depth{B}) \\
  \depth{A \cap B} &= \max(\depth{A}, \depth{B})
\end{align*}

\begin{theorem}[Transitivity of $<:$]\label{thm:⊑-trans}
    If $A <: B$ and $B <: C$, then $A <: C$.
\end{theorem}
\begin{proof}
  The proof is by well-founded induction on the relation $\ll$.
  We proceed by cases on the last rule applied in the
  derivation of $B <: C$.
  \begin{description}
  \item[Case (refl$_\alpha$)] We have $B = C = \alpha$.  From the premise $A <:
    B$ we immediately conclude that $A <: \alpha$.
  \item[Case (lb$_L$)] So $B = B_1 \cap B_2$, $A <: B_1 \cap
    B_2$, and $B_1 <: C$.
    We have $A <: B_1$ (Proposition~\ref{prop:subtyping} part
    \ref{prop:⊔⊑-inv}), so we conclude that $A <: C$ by the induction
    hypothesis, noting that $\langle A, B_1, C \rangle \ll \langle A,
    B, C \rangle$
    because $\depth{B_1} \le \depth{B}$
    and $\size{B_1} < \size{B}$.
  \item[Case (lb$_R$)] So $B = B_1 \cap B_2$, $A <:
    B_1 \cap B_2$, and $B_2 <: C$.  We have $A <: B_2$
    (Proposition~\ref{prop:subtyping} part \ref{prop:⊔⊑-inv}), so we
    conclude that $A <: C$ by the induction hypothesis, noting that
    $\langle A, B_2, C \rangle \ll \langle A, B, C \rangle$ because
    $\depth{B_2} \leq \depth{B}$ and $\size{B_2} < \size{B}$.
    
  \item[Case (glb)] We have $C = C_1 \cap C_2$, $B <: C_1$, and $B <:
    C_2$.  By the induction hypothesis, we have $A <: C_1$ and $A <:
    C_2$, noting that $\langle A, B, C_1 \rangle \ll \langle A, B, C
    \rangle$ and $\langle A, B, C_2 \rangle \ll \langle A, B, C\rangle$
    because $\size{C_1} < C$ and $\size{C_2} < C$.
    We conclude $A <: C_1 \cap C_2$ by rule (glb).
  \item[Case $(\to')$] So $C = C_1 \to C_2$, $\neg \topP{C_2}$, and
    there exists $B'$ such that $C_1 <: \dom{B'}$, $\cod{B'} <: C_2$,
    $B' \containedin B$, and $\neg \topInCod{B'}$. From $A <: B$ and $B' \containedin B$, we
    have $A <: B'$ (Proposition~\ref{prop:subtyping} part
    \ref{prop:subset-sub-sub}). Thus, for any $B_1 \to B_2 \inside B'$,
    $B_1 \to B_2$ factors $A$ (Proposition~\ref{prop:⊑-fun-inv}). We
    have satisfied the premises of Lemma~\ref{lem:sub-inv-trans}, so
    $\dom{B'} \to \cod{B'}$ factors $A$. That means there exists $A'$
    such that $A' \containedin A$, $\neg\,\topInCod{A'}$, $\dom{B'} <: \dom{A'}$,
    and $\cod{A'} <: \cod{B'}$. Then by the induction hypothesis, we
    have
    \[
    C_1 <: \dom{A'} \quad\text{and}\quad \cod{A'} <: C_2
    \]
    noting that
    \[
    \langle C_1, \dom{B'}, \dom{A'} \rangle \ll \langle A, B, C \rangle
    \]
    and
    \[
    \langle \cod{A'}, \cod{B'}, C_2 \rangle \ll \langle A, B, C \rangle
    \]
    because $\depth{\dom{B'}} < \depth{B}$
    and $\depth{\cod{B'}} < \depth{B}$.
    We conclude that $A <: C_1 \to C_2$ by rule $(\to')$ witnessed by $A'$.
  \item[Case $(\TOP_{\mathrm{top}})$] We have $C = \TOP$ and
    conclude $A <: \TOP$ by rule $(\TOP_{\mathrm{top}})$.
  \item[Case $(\TOP{\to}')$] We have $C = C_1 \to C_2$ and $\topP{C_2}$.
    We conclude $A <: C_1 \to C_2$ by rule $(\TOP{\to}')$.
  \end{description}
\end{proof}

\section{Equivalence with BCD Subtyping}
\label{sec:equiv}

Having proved (trans), we next prove $(\to)$ and $({\to}{\cap})$
and then show that $A <: B$ is equivalent to $A \leq B$.

\begin{lemma}[$\to$]\label{lem:⊑-fun′}
  If $C <: A$ and $B <: D$, then
  $A \to B <: C \to D$.
\end{lemma}
\begin{proof}
  Consider whether $\topP{D}$ or not.
  \begin{description}
  \item[Case $\topP{D}$] We conclude $A \to B <: C \to D$
    by rule $(\TOP{\to}')$.
  \item[Case $\neg \,\topP{D}$]
    Consider whether $\topP{B}$ or not.
    \begin{description}
    \item[Case $\topP{B}$] So $\topP{D}$ (Prop. \ref{prop:top}
      part \ref{prop:AllBot-⊑}), but that is a contradiction.
    \item[Case $\neg\,\topP{B}$]
      We conclude that $A \to B <: C \to D$ by rule $(\to')$.
    \end{description}
  \end{description}
\end{proof}

\begin{lemma}[${\to}{\cap}$]\label{lem:⊑-dist}
  $(A \to B) \cap (A \to C) <: A \to (B \cap C)$
\end{lemma}
\begin{proof}
   We consider the cases for whether $\topP{B}$ or $\topP{C}$.
   \begin{description}
   \item[Case $\topP{B}$ and $\topP{C}$]
     Then $\topP{B \cap C}$ and we conclude that
     $(A \to B) \cap (A \to C) <: A \to (B \cap C)$
     by rule $(\TOP{\to}')$.
   \item[Case $\topP{B}$ and $\neg\,\topP{C}$]
     We conclude that 
     $(A \to B) \cap (A \to C) <: A \to (B \cap C)$
     by rule $(\to')$, choosing the witness $A \to C$
     and noting that $C <: B$ 
     by way of Proposition~\ref{prop:top}
     part \ref{prop:AllBot-⊑-any}
     and $C <: C$ by Proposition~\ref{prop:subtyping}
     part \ref{prop:⊑-refl}.
   \item[Case $\neg\,\topP{B}$ and $\topP{C}$]
     We conclude that 
     $(A \to B) \cap (A \to C) <: A \to (B \cap C)$
     by rule $(\to')$, this time with witness $A \to B$
     and noting that
     $B <: B$ by Proposition~\ref{prop:subtyping}
     part \ref{prop:⊑-refl}
     and $B <: C$ 
     by way of Proposition~\ref{prop:top}
     part \ref{prop:AllBot-⊑-any}.
   \item[Case $\neg,\topP{B}$ and $\neg\,\topP{C}$]
     Again we apply rule $(\to')$, but with witness
     $(A \to B) \cap (A \to C)$.
   \end{description}
\end{proof}

\noindent We require one more lemma.

\begin{lemma}
  \label{lem:dv↦cv<:v}
  $A \leq \dom{A} \to \cod{A}$.
\end{lemma}
\begin{proof}
  The proof is by induction on $A$.
\end{proof}

\noindent Now for the proof of equivalence

\begin{theorem}[Equivalence of the subtyping relations]\ \\
  $A <: B$ if and only if $A \leq B$.
\end{theorem}
\begin{proof}
  We prove each direction of the if-and-only-if separately.
  \begin{description}
  \item[$A <: B$ implies $A \leq B$]
    We proceed by induction on the derivation of $A <: B$.
    \begin{description}
    \item[Case (refl$_\alpha$)] We conclude $\alpha \leq \alpha$ by (refl).      
    \item[Case (lb$_L$)] By the induction hypothesis we have $A \leq C$.
      By (incl$_L$) we have $A \cap B \leq A$. We conclude
      that $A \cap B \leq C$ by (trans).
    \item[Case (lb$_R$)] By the induction hypothesis we have $B \leq C$.
      By (incl$_R$) we have $A \cap B \leq B$. We conclude
      that $A \cap B \leq C$ by (trans).
    \item[Case (glb)] By the induction hypothesis we have
      $A \leq C$ and $A \leq D$, so we conclude that
      $A \leq C \cap D$ by (glb).
    \item[Case $(\to')$] By the induction hypothesis we have
      $C \leq \dom{B}$ and also $\cod{B} \leq D$.
      From $B \containedin A$ we have $A \leq B$.      
      Then by Lemma~\ref{lem:dv↦cv<:v} we have
      $B \leq \dom{B} \to \cod{B}$.
      Also, we have $\dom{B} \to \cod{B} \leq C \to D$ by rule $(\to)$.
      We conclude that $A \leq C \to D$ by chaining the three prior
      facts using (trans).
    \item[Case $(\TOP_{\mathrm{top}})$]
      We conclude that $A \leq \TOP$ by $(\TOP_{\mathrm{top}})$.
    \item[Case $(\TOP{\to}')$] We have $A \leq \TOP$ and $\TOP <: C \to
      \TOP$.  Also, $C \to \TOP \leq C \to D$ because $\TOP <: D$ follows
      from $\topP{D}$.  Thus, applying (trans) we conclude $A \leq C \to
      D$.
    \end{description}
    
  \item[$A \leq B$ implies $A <: B$]
    We proceed by induction on the derivation of $A \leq B$.
    \begin{description}
    \item[Case (refl)] We conclude $A <: A$ by Prop.~\ref{prop:subtyping}
      part \ref{prop:⊑-refl}.
    \item[Case (trans)]
      By the induction hypothesis, we have $A <: B$ and $B <: C$.
      We conclude that $A <: C$ by Theorem~\ref{thm:⊑-trans}.
    \item[Case (incl$_L$)] We have $A <: A$
      (Prop.~\ref{prop:subtyping} part \ref{prop:⊑-refl}),
      and therefore $A \cap B <: A$ by rule (lb$_L$).
    \item[Case (incl$_R$)] We have $B <: B$
      (Prop.~\ref{prop:subtyping} part \ref{prop:⊑-refl}),
      and therefore $A \cap B <: B$ by rule (lb$_R$).
    \item[Case (glb)] By the induction hypothesis, we have
      $A <: C$ and $A <: D$, so we conclude that $A <: C \cap D$ by (glb).
    \item[Case $(\to)$] By the induction hypothesis, we have $C <: A$
      and $B <: D$. We conclude that $A \to B <: C \to D$
      by Lemma~\ref{lem:⊑-fun′}.
    \item[Case $({\to}{\cap})$] We conclude that
      $(A \to B) \cap (A \to C) <: A \to (B \cap C)$
      by Lemma~\ref{lem:⊑-dist}.
    \item[Case $(\TOP_{\mathrm{top}})$]
      We conclude that $A <: \TOP$ by rule $(\TOP_{\mathrm{top}})$.
    \item[Case $(\TOP{\to})$]
      We have $\topP{\TOP}$, so $\TOP <: C \to \TOP$ by rule $(\TOP{\to}')$.
    \end{description}
    
  \end{description}
\end{proof}

\section{Conclusion}
\label{sec:conclude}

In this article we present a new subtyping relation $A <: B$ for
intersection types that satisfies a property similar to the subformula
property.  None of the rules of the new subtyping relation are
particularly novel, but the fact that we can prove transitivity
directly from is surprising!  We prove that the new relation is
equivalent to the subtyping relation $A \leq B$ of Barendregt, Coppo,
and Dezani-Ciancaglini.

\section*{Acknowledgments}

This material is based upon work supported by the National Science
Foundation under Grant No. 1814460.

\pagebreak

\bibliographystyle{plainnat}
\bibliography{all}

\begin{thebibliography}{48}
\providecommand{\natexlab}[1]{#1}
\providecommand{\url}[1]{\texttt{#1}}
\expandafter\ifx\csname urlstyle\endcsname\relax
  \providecommand{\doi}[1]{doi: #1}\else
  \providecommand{\doi}{doi: \begingroup \urlstyle{rm}\Url}\fi

\bibitem[Mic(2020)]{Microsoft:TypeScript2020aa}
Handbook: The typescript language reference, 2020.
\newblock URL
  \url{https://www.typescriptlang.org/docs/handbook/basic-types.html}.

\bibitem[Abramsky and Ong(1993)]{Abramsky:1993fk}
S.~Abramsky and C.H.L. Ong.
\newblock Full abstraction in the lazy lambda calculus.
\newblock \emph{Information and Computation}, 105\penalty0 (2):\penalty0 159 --
  267, 1993.
\newblock ISSN 0890-5401.
\newblock \doi{10.1006/inco.1993.1044}.
\newblock URL
  \url{http://www.sciencedirect.com/science/article/pii/S0890540183710448}.

\bibitem[Alessi et~al.(2006)Alessi, Barbanera, and
  Dezani-Ciancaglini]{Alessi:2006aa}
Fabio Alessi, Franco Barbanera, and Mariangiola Dezani-Ciancaglini.
\newblock Intersection types and lambda models.
\newblock \emph{Theoretical Compututer Science}, 355\penalty0 (2):\penalty0
  108--126, 2006.

\bibitem[Amin and Rompf(2017)]{Amin:2017aa}
Nada Amin and Tiark Rompf.
\newblock Type soundness proofs with definitional interpreters.
\newblock In \emph{Proceedings of the 44th ACM SIGPLAN Symposium on Principles
  of Programming Languages}, POPL 2017, pages 666--679, New York, NY, USA,
  2017. ACM.
\newblock ISBN 978-1-4503-4660-3.
\newblock \doi{10.1145/3009837.3009866}.
\newblock URL \url{http://doi.acm.org/10.1145/3009837.3009866}.

\bibitem[Amin et~al.(2016)Amin, Gr\"utter, Odersky, Rompf, and
  Stucki]{Amin:2016aa}
Nada Amin, Samuel Gr\"utter, Martin Odersky, Tiark Rompf, and Sandro Stucki.
\newblock The essence of dependent object types.
\newblock In \emph{A list of successes that can change the world}, LNCS.
  Springer, 2016.

\bibitem[Appel(1992)]{Appel:1992fk}
Andrew~W. Appel.
\newblock \emph{Compiling with continuations}.
\newblock Cambridge University Press, New York, NY, USA, 1992.
\newblock ISBN 0-521-41695-7.

\bibitem[Barendregt et~al.(1983)Barendregt, Coppo, and
  Dezani-Ciancaglini]{Barendregt:1983aa}
Henk Barendregt, Mario Coppo, and Mariangiola Dezani-Ciancaglini.
\newblock A filter lambda model and the completeness of type assignment.
\newblock \emph{Journal of Symbolic Logic}, 48\penalty0 (4):\penalty0 931--940,
  12 1983.
\newblock \doi{10.2307/2273659}.

\bibitem[Barendregt et~al.(2013)Barendregt, Dekkers, and
  Statman]{Barendregt:2013aa}
Henk Barendregt, Wil Dekkers, and Richard Statman.
\newblock \emph{Lambda Calculus with Types}.
\newblock Perspectives in Logic. Cambridge University Press, 2013.

\bibitem[Benton et~al.(2009)Benton, Kennedy, and Varming]{Benton:2009ab}
Nick Benton, Andrew Kennedy, and Carsten Varming.
\newblock \emph{Some Domain Theory and Denotational Semantics in Coq}, pages
  115--130.
\newblock Springer Berlin Heidelberg, Berlin, Heidelberg, 2009.
\newblock ISBN 978-3-642-03359-9.
\newblock \doi{10.1007/978-3-642-03359-9_10}.
\newblock URL \url{http://dx.doi.org/10.1007/978-3-642-03359-9_10}.

\bibitem[Bi et~al.(2018)Bi, d.~S.~Oliveira, and Schrijvers]{Bi:2018aa}
Xuan Bi, Bruno~C. d.~S.~Oliveira, and Tom Schrijvers.
\newblock {The Essence of Nested Composition}.
\newblock In Todd Millstein, editor, \emph{32nd European Conference on
  Object-Oriented Programming (ECOOP 2018)}, volume 109 of \emph{Leibniz
  International Proceedings in Informatics (LIPIcs)}, pages 22:1--22:33,
  Dagstuhl, Germany, 2018. Schloss Dagstuhl--Leibniz-Zentrum fuer Informatik.
\newblock ISBN 978-3-95977-079-8.
\newblock \doi{10.4230/LIPIcs.ECOOP.2018.22}.
\newblock URL \url{http://drops.dagstuhl.de/opus/volltexte/2018/9227}.

\bibitem[Bi et~al.(2019)Bi, Xie, Oliveira, and Schrijvers]{Bi:2019aa}
Xuan Bi, Ningning Xie, Bruno C. d.~S. Oliveira, and Tom Schrijvers.
\newblock Distributive disjoint polymorphism for compositional programming.
\newblock In Lu{\'\i}s Caires, editor, \emph{Programming Languages and
  Systems}, pages 381--409, Cham, 2019. Springer International Publishing.
\newblock ISBN 978-3-030-17184-1.

\bibitem[Castagna et~al.(2014)Castagna, Nguyen, Xu, Im, Lenglet, and
  Padovani]{Castagna:2014aa}
Giuseppe Castagna, Kim Nguyen, Zhiwu Xu, Hyeonseung Im, Sergue\"{\i} Lenglet,
  and Luca Padovani.
\newblock Polymorphic functions with set-theoretic types: Part 1: Syntax,
  semantics, and evaluation.
\newblock In \emph{Symposium on Principles of Programming Languages}, POPL,
  pages 5--17. ACM, 2014.

\bibitem[Chaudhuri()]{Chaudhuri:2014aa}
Avik Chaudhuri.
\newblock Flow: a static type checker for javascript.
\newblock URL \url{http://flowtype.org/}.

\bibitem[Coppo and Dezani-Ciancaglini(1980)]{Coppo:1980ab}
M.~Coppo and M.~Dezani-Ciancaglini.
\newblock An extension of the basic functionality theory for the
  $\lambda$-calculus.
\newblock \emph{Notre Dame J. Formal Logic}, 21\penalty0 (4):\penalty0
  685--693, 10 1980.
\newblock \doi{10.1305/ndjfl/1093883253}.
\newblock URL \url{https://doi.org/10.1305/ndjfl/1093883253}.

\bibitem[Coppo et~al.(1979)Coppo, Dezani-Ciancaglini, and Salle']{Coppo:1979aa}
M.~Coppo, M.~Dezani-Ciancaglini, and P.~Salle'.
\newblock \emph{Functional characterization of some semantic equalities inside
  $\lambda$-calculus}, pages 133--146.
\newblock Springer Berlin Heidelberg, Berlin, Heidelberg, 1979.

\bibitem[Coppo et~al.(1981)Coppo, Dezani-Ciancaglini, and
  Venneri]{Coppo:1981aa}
M.~Coppo, M.~Dezani-Ciancaglini, and B.~Venneri.
\newblock Functional characters of solvable terms.
\newblock \emph{Mathematical Logic Quarterly}, 27\penalty0 (2-6):\penalty0
  45--58, 1981.
\newblock ISSN 1521-3870.
\newblock \doi{10.1002/malq.19810270205}.
\newblock URL \url{http://dx.doi.org/10.1002/malq.19810270205}.

\bibitem[Coppo et~al.(1984)Coppo, Dezani-Ciancaglini, Honsell, and
  Longo]{Coppo:1984aa}
M.~Coppo, M.~Dezani-Ciancaglini, F.~Honsell, and G.~Longo.
\newblock Extended type structures and filter lambda models.
\newblock In G.~Longo G.~Lolli and A.~Marcja, editors, \emph{Logic Colloquium
  '82}, volume 112 of \emph{Studies in Logic and the Foundations of
  Mathematics}, pages 241 -- 262. Elsevier, 1984.
\newblock \doi{http://dx.doi.org/10.1016/S0049-237X(08)71819-6}.
\newblock URL
  \url{http://www.sciencedirect.com/science/article/pii/S0049237X08718196}.

\bibitem[Dezani-Ciancaglini et~al.(2005)Dezani-Ciancaglini, Honsell, and
  Motohama]{Dezani-Ciancaglini:2005aa}
M.~Dezani-Ciancaglini, F.~Honsell, and Y.~Motohama.
\newblock Compositional characterisations of $\lambda$-terms using intersection
  types.
\newblock \emph{Theoretical Computer Science}, 340\penalty0 (3):\penalty0 459
  -- 495, 2005.
\newblock ISSN 0304-3975.
\newblock \doi{https://doi.org/10.1016/j.tcs.2005.03.011}.
\newblock URL
  \url{http://www.sciencedirect.com/science/article/pii/S0304397505001325}.
\newblock Mathematical Foundations of Computer Science 2000.

\bibitem[Dockins(2014)]{Dockins:2014aa}
Robert Dockins.
\newblock Formalized, effective domain theory in coq.
\newblock In \emph{Interactive Theorem Proving}, ITP, 2014.

\bibitem[Downey et~al.(1980)Downey, Sethi, and Tarjan]{Downey:JACM:1980}
Peter~J. Downey, Ravi Sethi, and Robert~Endre Tarjan.
\newblock Variations on the common subexpression problem.
\newblock \emph{JACM}, 27\penalty0 (4):\penalty0 758--771, 1980.
\newblock ISSN 0004-5411.

\bibitem[Dunfield and Krishnaswami(2019)]{Dunfield:2019aa}
Joshua Dunfield and Neel Krishnaswami.
\newblock Bidirectional typing, 2019.

\bibitem[Engeler(1981)]{Engeler:1981aa}
Erwin Engeler.
\newblock Algebras and combinators.
\newblock \emph{Algebra Universalis}, 13\penalty0 (1):\penalty0 389--392, Dec
  1981.

\bibitem[Felleisen et~al.(2009)Felleisen, Findler, and Flatt]{Felleisen:2009aa}
Matthias Felleisen, Robert~Bruce Findler, and Matthew Flatt.
\newblock \emph{Semantics Engineering with PLT Redex}.
\newblock MIT Press, 2009.

\bibitem[Honsell and Lenisa(1999)]{Honsell:1999aa}
Furio Honsell and Marina Lenisa.
\newblock Semantical analysis of perpetual strategies in λ-calculus.
\newblock \emph{Theoretical Computer Science}, 212\penalty0 (1):\penalty0 183
  -- 209, 1999.
\newblock ISSN 0304-3975.
\newblock \doi{https://doi.org/10.1016/S0304-3975(98)00140-6}.
\newblock URL
  \url{http://www.sciencedirect.com/science/article/pii/S0304397598001406}.

\bibitem[Honsell and Rocca(1992)]{Honsell:1992aa}
Furio Honsell and Simonetta Ronchi~Della Rocca.
\newblock An approximation theorem for topological lambda models and the
  topological incompleteness of lambda calculus.
\newblock \emph{Journal of Computer and System Sciences}, 45\penalty0
  (1):\penalty0 49 -- 75, 1992.
\newblock ISSN 0022-0000.
\newblock \doi{https://doi.org/10.1016/0022-0000(92)90040-P}.
\newblock URL
  \url{http://www.sciencedirect.com/science/article/pii/002200009290040P}.

\bibitem[Ishihara and Kurata(2002)]{Ishihara:2002aa}
Hajime Ishihara and Toshihiko Kurata.
\newblock Completeness of intersection and union type assignment systems for
  call-by-value $\lambda$-models.
\newblock \emph{Theoretical Computer Science}, 272\penalty0 (1):\penalty0 197
  -- 221, 2002.
\newblock ISSN 0304-3975.
\newblock \doi{https://doi.org/10.1016/S0304-3975(00)00351-0}.
\newblock URL
  \url{http://www.sciencedirect.com/science/article/pii/S0304397500003510}.
\newblock Theories of Types and Proofs 1997.

\bibitem[Kumar et~al.(2014)Kumar, Myreen, Norrish, and Owens]{Kumar:2014aa}
Ramana Kumar, Magnus~O. Myreen, Michael Norrish, and Scott Owens.
\newblock {CakeML}: A verified implementation of {ML}.
\newblock In \emph{Principles of Programming Languages ({POPL})}, pages
  179--191. ACM Press, January 2014.
\newblock \doi{10.1145/2535838.2535841}.
\newblock URL \url{https://cakeml.org/popl14.pdf}.

\bibitem[Landin(1966)]{Landin:1966la}
P.~J. Landin.
\newblock The next 700 programming languages.
\newblock \emph{Commun. ACM}, 9\penalty0 (3):\penalty0 157--166, 1966.
\newblock ISSN 0001-0782.

\bibitem[Laurent(2018)]{Laurent:2018aa}
Olivier Laurent.
\newblock Intersection subtyping with constructors.
\newblock In Michele Pagani, editor, \emph{Proceedings of the Ninth Workshop on
  {I}ntersection {T}ypes and {R}elated {S}ystems}, July 2018.

\bibitem[Mossin(2003)]{Mossin:2003aa}
Christian Mossin.
\newblock Exact flow analysis.
\newblock \emph{Mathematical Structures in Computer Science}, 13\penalty0
  (1):\penalty0 125--156, February 2003.

\bibitem[Muehlboeck and Tate(2018)]{Muehlboeck:2018aa}
Fabian Muehlboeck and Ross Tate.
\newblock Empowering union and intersection types with integrated subtyping.
\newblock \emph{Proc. ACM Program. Lang.}, 2\penalty0 (OOPSLA):\penalty0
  112:1--112:29, October 2018.
\newblock ISSN 2475-1421.
\newblock \doi{10.1145/3276482}.
\newblock URL \url{http://doi.acm.org/10.1145/3276482}.

\bibitem[Odersky et~al.(2004)Odersky, Altherr, Cremet, Dragos, Dubochet, Emir,
  McDirmid, Micheloud, Mihaylov, Schinz, Stenmn, Spoon, and
  Zenger]{Odersky:2004aa}
Martin Odersky, Philippe Altherr, Vincent Cremet, Iulian Dragos, Gilles
  Dubochet, Burak Emir, Sean McDirmid, Stephane Micheloud, Nikolay Mihaylov,
  Michel Schinz, Erik Stenmn, Lex Spoon, and Matthias Zenger.
\newblock An overview of the {Scala} programming language.
\newblock Technical Report IC/2004/64, EPFL, 2004.

\bibitem[Oliveira et~al.(2016)Oliveira, Shi, and Alpuim]{Oliveira:2016aa}
Bruno C. d.~S. Oliveira, Zhiyuan Shi, and Jo\~{a}o Alpuim.
\newblock Disjoint intersection types.
\newblock In \emph{Proceedings of the 21st ACM SIGPLAN International Conference
  on Functional Programming}, ICFP 2016, pages 364--377, New York, NY, USA,
  2016. Association for Computing Machinery.
\newblock ISBN 9781450342193.
\newblock \doi{10.1145/2951913.2951945}.
\newblock URL \url{https://doi.org/10.1145/2951913.2951945}.

\bibitem[Owens et~al.(2017)Owens, Norrish, Kumar, Myreen, and
  Tan]{Owens:2017aa}
Scott Owens, Michael Norrish, Ramana Kumar, Magnus~O. Myreen, and Yong~Kiam
  Tan.
\newblock Verifying efficient function calls in {CakeML}.
\newblock \emph{Proc. ACM Program. Lang.}, 1\penalty0 (ICFP), September 2017.
\newblock \doi{10.1145/3110262}.
\newblock URL \url{https://cakeml.org/icfp17.pdf}.

\bibitem[Palsberg and Pavlopoulou(1998)]{Palsberg:1998aa}
Jens Palsberg and Christina Pavlopoulou.
\newblock From polyvariant flow information to intersection and union types.
\newblock In \emph{Proceedings of the 25th ACM SIGPLAN-SIGACT Symposium on
  Principles of Programming Languages}, POPL '98, pages 197--208, New York, NY,
  USA, 1998. Association for Computing Machinery.
\newblock ISBN 0897919793.
\newblock \doi{10.1145/268946.268963}.
\newblock URL \url{https://doi.org/10.1145/268946.268963}.

\bibitem[Pierce(1989)]{Pierce:1989aa}
Benjamin C~. Pierce.
\newblock A decision procedure for the subtype relation on intersection types
  with bounded variables.
\newblock Technical Report CMU-CS-89-169, Carnegie Mellon University, 1989.

\bibitem[Pierce(1991)]{Pierce:1991aa}
Benjamin~C. Pierce.
\newblock \emph{Programming with Intersection Types and Bounded Polymorphism}.
\newblock PhD thesis, Carnegie Mellon University, 1991.

\bibitem[Plotkin(1975)]{G.-D.-Plotkin:1975on}
G.~D. Plotkin.
\newblock Call-by-name, call-by-value and the lambda-calculus.
\newblock \emph{Theoretical Computer Science}, 1\penalty0 (2):\penalty0
  125--159, December 1975.

\bibitem[Plotkin(1993)]{Plotkin:1993ab}
Gordon~D. Plotkin.
\newblock Set-theoretical and other elementary models of the λ-calculus.
\newblock \emph{Theoretical Computer Science}, 121\penalty0 (1):\penalty0 351
  -- 409, 1993.

\bibitem[Reynolds(1988)]{Reynolds:1988aa}
John~C. Reynolds.
\newblock Preliminary design of the programming language forsythe.
\newblock Technical Report CMU-CS-88-159, Carnegie Mellon University, Computer
  Science Dept., June 1988.

\bibitem[Rompf and Amin(2016)]{Rompf:2016aa}
Tiark Rompf and Nada Amin.
\newblock Type soundness for dependent object types (dot).
\newblock In \emph{Proceedings of the 2016 ACM SIGPLAN International Conference
  on Object-Oriented Programming, Systems, Languages, and Applications}, OOPSLA
  2016, pages 624--641, New York, NY, USA, 2016. Association for Computing
  Machinery.
\newblock ISBN 9781450344449.
\newblock \doi{10.1145/2983990.2984008}.
\newblock URL \url{https://doi.org/10.1145/2983990.2984008}.

\bibitem[{Ronchi Della Rocca} and Paolini(2004)]{Rocca:2004aa}
Simona {Ronchi Della Rocca} and Luca Paolini.
\newblock \emph{The Parametric Lambda Calculus}.
\newblock Springer, 2004.

\bibitem[Scott(1976)]{Scott:1976lq}
Dana Scott.
\newblock Data types as lattices.
\newblock \emph{SIAM Journal on Computing}, 5\penalty0 (3):\penalty0 522--587,
  1976.

\bibitem[Siek(2018)]{Siek:2018aa}
Jeremy~G. Siek.
\newblock Intersection types, sub-formula property, and the functional
  character of the lambda calculus, August 2018.
\newblock URL \url{http://siek.blogspot.com/2018/08/}.

\bibitem[Sim{\~o}es et~al.(2007)Sim{\~o}es, Hammond, Florido, and
  Vasconcelos]{Simoes:2007aa}
Hugo~R. Sim{\~o}es, Kevin Hammond, M{\'a}rio Florido, and Pedro Vasconcelos.
\newblock Using intersection types for cost-analysis of higher-order
  polymorphic functional programs.
\newblock In Thorsten Altenkirch and Conor McBride, editors, \emph{Types for
  Proofs and Programs}, pages 221--236, Berlin, Heidelberg, 2007. Springer
  Berlin Heidelberg.
\newblock ISBN 978-3-540-74464-1.

\bibitem[Tarditi et~al.(1996)Tarditi, Morrisett, Cheng, Stone, Harper, and
  Lee]{Tarditi:1996aa}
D.~Tarditi, G.~Morrisett, P.~Cheng, C.~Stone, R.~Harper, and P.~Lee.
\newblock Til: A type-directed optimizing compiler for ml.
\newblock In \emph{Proceedings of the ACM SIGPLAN 1996 Conference on
  Programming Language Design and Implementation}, PLDI '96, pages 181--192,
  New York, NY, USA, 1996. Association for Computing Machinery.
\newblock ISBN 0897917952.
\newblock \doi{10.1145/231379.231414}.
\newblock URL \url{https://doi.org/10.1145/231379.231414}.

\bibitem[Turbak et~al.(1997)Turbak, Dimock, Muller, and Wells]{Turbak:1997aa}
Franklyn Turbak, Allyn Dimock, Robert Muller, and J.~B. Wells.
\newblock Compiling with polymorphic and polyvariant flow types.
\newblock In \emph{Workshop on Types in Compilation}, 1997.

\bibitem[van Bakel(1995)]{Bakel:1995aa}
Steffen van Bakel.
\newblock Intersection type assignment systems.
\newblock \emph{Theoretical Computer Science}, 151\penalty0 (2):\penalty0 385
  -- 435, 1995.
\newblock ISSN 0304-3975.
\newblock \doi{https://doi.org/10.1016/0304-3975(95)00073-6}.
\newblock URL
  \url{http://www.sciencedirect.com/science/article/pii/0304397595000736}.
\newblock 13th Conference on Foundations of Software Technology and Theoretical
  Computer Science.

\end{thebibliography}

\end{document}